\documentclass{new_tlp}
\usepackage{amsmath}
\usepackage{amssymb}
\usepackage{mathtools}
\usepackage{natbib}
\usepackage{tikz}
\usepackage{caption}
\usetikzlibrary{positioning, automata}

\tikzset{place/.style={circle}}
\tikzset{>=stealth, auto, node distance=2.5cm, every loop/.style={->, min distance=10mm, in=0, out=60, looseness=10}}

\def\bp{\textbf{p}}
\usepackage{tikz}

\def\rar{\rightarrow}
\def\lrar{\leftrightarrow}
\def\beq{\begin{equation}}
\def\eeq#1{\label{#1}\end{equation}}
\def\ba{\begin{array}}
\def\ea{\end{array}}

\def\mc{\mathcal{C}}
\def\mp{\mathcal{P}}
\def\mb{\mathcal{B}}
\def\ma{\mathcal{A}}

\def\mq{\mathcal{Q}}
\def\me{\mathcal{E}}
\def\mo{\mathcal{O}}
\def\bp{{\bf p}}

\newtheorem{theorem}{Theorem}

\newtheorem{lemma}{Lemma}
\newtheorem{prop}{Proposition}

\newtheorem{define}{Definition}

\title[Infinitary Formulas with Extensional Atoms]{\bf 
Stable Models for Infinitary Formulas\\
with Extensional Atoms} 
\author[A. Harrison and V. Lifschitz]{AMELIA HARRISON, VLADIMIR LIFSCHITZ\\  
University of Texas, Austin, Texas, USA \\
\email{ameliaj,vl@cs.utexas.edu}\\
}

\begin{document}

\date{}
\maketitle

\begin{abstract}
The definition of stable models for propositional formulas with infinite 
conjunctions
and disjunctions can be used to describe the semantics of answer set programming
languages.  In this note, we enhance that definition by  
introducing a distinction between intensional and extensional
atoms.  The symmetric splitting theorem for 
first-order formulas is then extended to
infinitary formulas and used to reason about infinitary definitions. This note 
is under consideration for publication in Theory and Practice of Logic 
Programming.  
\end{abstract}

\section{Introduction}

The original definition of a stable model \citep{gel88} was applicable only to 
quantifier-free formulas of a restricted syntax.  
Stable models for arbitrary first-order sentences were defined by \cite{fer07a} 
using the stable model operator SM. This definition can be used to define the 
semantics of some rules with aggregate expressions. For instance, the following 
rule, written in the input language of the ASP system {\sc 
clingo},\footnote{\sc http://potassco.sourceforge.net}
\beq
{\text {\tt q  :- \#count\{X:p(X)\} = 0}} 
\eeq{eq:r1}
can be identified with the first-order sentence 
\beq
\forall x \; \neg p(x) \rar q.
\eeq{eq:r1fo} 

In \cite{fer09}, that definition was generalized to  
allow a distinction 
between ``extensional'' and ``intensional'' predicate symbols. (Under the 
original definition all predicate symbols are treated as intensional.) 
Intuitively, an intensional predicate is one 
whose extent is defined by the program, while all other, extensional, 
predicates are defined externally. Similar 
distinctions have been proposed many 
times: \cite{DBLP:journals/ijseke/GelfondP96} distinguish between input and 
output predicates in their ``lp-functions'',
\cite{oik08} distinguish between input and output 
atoms, and \cite{lier11} between input and non-input atoms. These distinctions 
are useful because they 
allow for a modular view of logic programs. 
For example, in the splitting theorem from \cite{fer09a}, the authors showed 
that stable models for a program can sometimes be computed by breaking 
the program into parts and computing the stable models of each part separately 
using different sets of intensional predicates.

Using the approach proposed by \cite{fer05}, \cite{tru12} extended the 
definition of a stable model in a different direction: he showed how to apply 
this concept to infinitary propositional formulas. He also showed that the 
definition of first-order stable models in terms of the 2007 definition of the 
operator SM could be reduced to the 
definition of infinitary stable models. Infinitary 
stable models were used in that paper as a tool for relating first-order stable 
models to the semantics of first-order logic with inductive definitions.
Infinitary stable models are important also because they 
provide an alternative understanding of the semantics of aggregates. For 
instance, rule \eqref{eq:r1} can be identified with the infinitary formula
\beq
\bigwedge_{t} \neg p(t) \rar q,  
\eeq{eq:r1if}
where the conjunction in the antecedent is understood as ranging over all 
ground terms~$t$ not containing arithmetic operations. The advantage of this 
approach over the use of first-order formulas is that it is more flexible. For 
example, it is applicable to 
aggregates involving {\tt \#sum}. In recent work, \cite{geb15} use this idea 
to define a precise semantics for a large class of ASP programs, 
including programs with local variables and aggregate expressions. 

However, Truszczynski's definition of stable models for infinitary formulas 
does not allow a distinction between 
extensional and intensional atoms. It treats all atoms as intensional.  In 
this note, we generalize the definition of  stable models for infinitary 
formulas to accommodate both intensional and  extensional 
atoms, and we study properties of this definition. As might be expected, 
the definition of first-order stable models with extensional predicates can be 
reduced to the definition proposed in this note. We use this definition to 
generalize the results on first-order splitting from \cite{fer09a}. 
In particular, we look at the splitting lemma from  \cite{fer09a}, which showed  
that  under certain conditions the stable models of a formula can be computed by 
computing the stable models of the same formula with respect to smaller sets 
of intensional predicates. We find that a straightforward infinitary counterpart 
to the splitting lemma does not hold, and show how the lemma needs to be 
modified for the infinitary case. The situation is similar for the splitting 
theorem discussed above. The infinitary splitting theorem is used to generalize 
the lemma on explicit definitions due to \cite{fer05}, which describes how 
adding explicit definitions to a program affects its stable models. In the 
version presented in this note, the program can include infinitary formulas 
and the definition can be recursive. 

\section{Review: Infinitary Formulas and their Stable Models}\label{sec:ifandsm} 

This review follows \cite{tru12}, \cite{har15a}.
Let $\sigma$ be a propositional signature,
that is, a set of propositional atoms.  For every nonnegative integer~$r$,
{\it (infinitary propositional) formulas (over $\sigma$) of rank~$r$} are
defined recursively, as follows:
\begin{itemize}
\item every atom from~$\sigma$ is a formula of rank~0,
\item if $\mathcal{H}$ is a set of formulas, and~$r$ is the smallest
nonnegative 
integer that is greater than the ranks of all elements of $\mathcal{H}$,
then $\mathcal{H}^\land$ and $\mathcal{H}^\lor$ are formulas of rank~$r$,
\item if $F$ and $G$ are formulas, and~$r$ is the smallest nonnegative
integer that is greater than the ranks of~$F$ and~$G$, then $F\rar G$ is a
formula of rank~$r$.
\end{itemize}
We will write $\{F,G\}^\land$ as $F\land G$, and
$\{F,G\}^\lor$ as $F\lor G$.
The symbols $\top$ and $\bot$ will be understood as abbreviations 
for~$\emptyset^{\land}$ and $\emptyset^{\lor}$ respectively; 
$\neg F$ stands for $F\rar\bot$, and $F\lrar G$ 
stands for \hbox{$(F\rar G)\land(G\rar F)$}.
These conventions allow us to view
finite propositional formulas over~$\sigma$
as a special case of infinitary formulas.

A set or family of formulas is {\it bounded} if the ranks of its members
are bounded from above.  For any bounded family $(F_\alpha)_{\alpha \in A}$
of formulas, the formula $\{F_\alpha:~{\alpha~\in~A}\}^\land$ will be denoted 
by $\bigwedge_{\alpha \in A} F_\alpha$, and similarly for disjunctions. 

Subsets of a signature~$\sigma$ will be also called {\it interpretations}
of~$\sigma$.
The satisfaction relation between an interpretation and a formula is
defined recursively, as follows:
\begin{itemize}
\item For every atom $p$ from $\sigma$, $I\models p$ if $p\in I$.
\item $I\models\mathcal{H}^\land$ if for every formula $F$ in~$\mathcal{H}$,
$I\models F$.
\item $I\models\mathcal{H}^\lor$ if there is a formula $F$ in~$\mathcal{H}$
such that $I\models F$.
\item $I\models F\rar G$ if $I\not\models F$ or $I\models G$.
\end{itemize}
An infinitary formula is {\it tautological} if it is satisfied by all 
interpretations. Two infinitary formulas are {\it equivalent} if they are 
satisfied by the same interpretations.

The {\it reduct} $F^I$ of a formula~$F$ w.r.t.~an interpretation~$I$ is
defined recursively, as follows:
\begin{itemize}
\item For every atom~$p$ from $\sigma$, $p^I$ is~$p$ if $p\in I$, and $\bot$
otherwise.
\item $(\mathcal{H}^\land)^I$ is $\{G^I\ |\ G\in\mathcal{H}\}^\land$.
\item $(\mathcal{H}^\lor)^I$ is $\{G^I\ |\ G\in\mathcal{H}\}^\lor$.
\item $(G\rar H)^I$ is $G^I\rar H^I$ if $I\models G\rar H$, and $\bot$
otherwise.
\end{itemize}
If $\mathcal{H}$ is a set of infinitary formulas then the {\it reduct} 
$\mathcal{H}^I$ is the set $\{F^I: F \in \mathcal{H}\}$. 
An interpretation~$I$ is a {\it stable model} of a set $\mathcal{H}$ of
formulas if it is minimal w.r.t.~set inclusion among the interpretations
satisfying the reduct $\mathcal{H}^I$.

\medskip

\noindent {\it Example}
\newline \noindent It is clear that $\{q\}$ is the only stable model of 
\eqref{eq:r1if}. Indeed, the reduct of \eqref{eq:r1if} w.r.t. $\{q\}$ 
is  
\beq
\top \rar q, 
\eeq{eq:qred}
and $\{q\}$ is a minimal model of this formula w.r.t. set inclusion. It 
is easy to see that \eqref{eq:r1if} has no other stable models.  

\section{$\ma$-stable Models}\label{sec:astable}
\label{ex:running}
Following \cite{fer09}, we will 
assume that some atoms in a program are designated ``intensional'' while all 
others are regarded as ``extensional''. 

Recall that $\sigma$ denotes a propositional signature. Let $\mathcal{A} 
\subseteq \sigma$ be a 
(possibly infinite) set of atoms. The partial order $\leq_\ma$ is defined as 
follows: for any sets $I,J \subseteq \sigma$, we say that $I \leq_\ma J$ if 
$I \subseteq J$ and $J \setminus I \subseteq \ma$. (Intuitively, if the atoms 
in $\ma$ are treated as intensional and all other atoms from $\sigma$ are 
treated as extensional, the relation holds if $I \subseteq J$ and~$I,J$ 
agree on all extensional atoms.) 
An interpretation $I$ is called an 
(infinitary) $\mathcal{A}${\it-stable model} of a formula $F$ if it is a minimal 
model of $F^I$ w.r.t. $\leq_\ma$.

Observe that if $\ma = \sigma$ then $\ma$-stable models of a formula $F$ are 
the same as stable models. If $\ma = \emptyset$ then $\ma$-stable models are 
all models of $F$. Truszczynski observed that an interpretation 
$I$ satisfies $F$ iff $I$ satisfies $F^I$ \cite[Proposition 1]{tru12}. It 
follows that all $\ma$-stable models of $F$ also satisfy $F$. 

\medskip

\noindent{\it Example (continued)}
\newline \noindent To illustrate the definition of $\ma$-stability, let's find 
all $\{q\}$-stable 
models\footnote{ Here, we understand $\sigma$ as implicitly defined to be the 
set containing $q$ and all atoms of the form~$p(t)$ where $t$ is a ground 
term.} of \eqref{eq:r1if}. 
The stable model $\{q\}$ of \eqref{eq:r1if} is $\{q\}$-stable as well, because 
it is a minimal model of \eqref{eq:qred} w.r.t. $\leq_{\{q\}}$. On the 
other hand, any non-empty set $\mp$ of atoms of the form~$p(t)$ is 
$\{q\}$-stable too. Indeed, the reduct of \eqref{eq:r1if} w.r.t. 
such a set is an implication whose antecedent has $\bot$ as one of its 
conjunctive terms. Such a formula is tautological so that it is satisfied by 
$\mp$. Furthermore, $\mp$ is a minimal model w.r.t. $\leq_{\{q\}}$
since any subset of $\mp$ will disagree with it on extensional atoms. 

\medskip
  
The fact that all stable models of \eqref{eq:r1if} are also $\{q\}$-stable is an 
instance of a more general fact: If $I$ is an $\ma$-stable model of $F$ and 
$\mb$ is a subset of $\ma$ then $I$ is also a $\mb$-stable model of~$F$. 
This follows directly from the definition of $\ma$-stability.

The following proposition provides two alternative definitions for 
$\ma$-stability. 
\begin{prop}\label{thm:astable}
The following three conditions are equivalent:
\begin{enumerate}
\item[(i)] $I$ is an $\ma$-stable model of $F$;
\item[(ii)] $I$ is a minimal model (w.r.t. set inclusion) of 
\beq
F^I \land \bigwedge_{p \in I \setminus \mathcal{A}} p;
\eeq{eq:modred}
\item[(iii)] $I$ is a stable model of 
\beq
F \;\; \land \;\; \bigwedge_{p \in \sigma \setminus \mathcal{A}} (p \lor \neg 
p). 
\eeq{eq:extchoice}
\end{enumerate}
\end{prop}
\begin{proof}
We first establish that conditions (i) and (ii) are equivalent:
$I$ is an $\ma$-stable model of~$F$ 
\begin{align*}
\text{iff } & I \text{ is a minimal model of $F^I$ 
w.r.t. $\leq_\ma$}\\ 
\text{iff } & I \models F^I \text{ and there is no $J \subset I$ such that 
$J \models F^I$ and } I \setminus J \subseteq \ma \\
\text{iff } & I \models F^I \text{ and there is no $J \subset I$ such that $J 
\models F^I$ and } \forall p (p \in I \land p \not \in J \rar p \in \ma) \\
\text{iff } & I \models F^I \text{ and there is no $J \subset I$ such that $J 
\models F^I$ and } \forall p (p \in I \land p \not \in \ma \rar p \in J) \\
\text{iff } & I \models F^I \text{ and there is no $J \subset I$ such that $J 
\models F^I$ and } I \setminus \ma \subseteq J\\
\text{iff } & I \models F^I \land \bigwedge_{p \in I \setminus \ma} p  \text{ 
and there is no $J \subset I$ such that } 
J \models F^I \land \bigwedge_{p \in I \setminus \ma} p\\
\text{iff } & I \text{ is an minimal model of  \eqref{eq:modred}.}
\end{align*} 

Finally, we will establish that conditions (ii) and (iii) are equivalent. It is 
easy to see that the reduct of \eqref{eq:extchoice} is equivalent to 
\eqref{eq:modred}: 
\begin{align*}
& F^I \;\; \land \;\; \left ( \bigwedge_{p \in \sigma \setminus \mathcal{A}} (p \lor \neg 
p)\right )^I  \\
\lrar \;\;\;\;\; & F^I \;\; \land \;\;  \bigwedge_{p \in \sigma \setminus \mathcal{A}} (p^I \lor (\neg 
p)^I)  \\
\lrar \;\;\;\;\; & F^I \;\; \land \;\;  \bigwedge_{p \in I \setminus \mathcal{A}} (p^I \lor (\neg 
p)^I) \;\; \land \;\; 
\bigwedge_{p \in \sigma \setminus (I \cup \mathcal{A})} (p^I \lor (\neg p)^I) 
 \\
\lrar \;\;\;\;\; & F^I \;\; \land \;\;  \bigwedge_{p \in I \setminus \mathcal{A}} (p \lor 
\bot) \;\; \land \;\; 
\bigwedge_{p \in \sigma \setminus (I \cup \mathcal{A})} (\bot \lor \top) 
 \\
\lrar \;\;\;\;\; & F^I \;\; \land \;\;  \bigwedge_{p \in I \setminus \mathcal{A}} p . \\
\end{align*} 
So $I$ is a minimal model of \eqref{eq:modred} iff it is a 
stable model of \eqref{eq:extchoice}. 
\end{proof}

\section{Relating Infinitary and First-Order $\ma$-Stable Models}

As mentioned in the introduction, \cite{tru12} showed that infinitary stable models can be viewed as 
a generalization of first-order stable models in the sense of \cite{fer09}. In 
this section, we will show that the corresponding result holds for $\bp$-stable 
models as well.\footnote{The definition of ${\bp}$-stable models, where $\bp$ is 
a list of distinct predicate symbols, can be found in \cite{fer09}, Section 
2.3.}  First, we review Truszczynski's results.   

Let $\Sigma$ be a first-order signature, and $I$ be an interpretation of 
$\Sigma$ with non-empty domain~$|I|$. For each element~$u$ of~$|I|$, by~$u^*$ we 
denote a new object constant, called the {\it name of~$u$}. By~$\Sigma^{|I|}$ we 
denote the signature obtained by adding the names of all elements of~$|I|$ 
to~$\Sigma$. An interpretation $I$ is identified with its extension~$I'$ 
to~$\Sigma^{|I|}$ in which for each~$u$ in~$|I|$,~$I'(u^*) = u$.  
By~$A_{\Sigma, I}$ we denote the set of all atomic sentences over 
$\Sigma^{|I|}$ built with relation symbols from~$\Sigma$ and names of elements in 
$|I|$, and by $I^r$ we denote the subset of $A_{\Sigma, I}$ that describes in 
the obvious way the extents of the relations in $I$.  Let $F$ be a 
formula over signature~$\Sigma^{|I|}$. Then the 
{\it grounding of~$F$ w.r.t. 
$I$, gr$_I(F)$} is defined recursively, as follows:
\begin{itemize}
\item gr$_I(\bot)$ is $\bot$;
\item gr$_I(p(t_1, \dots, t_k))$ is $p((t_1^I)^*, \dots, (t_k^I)^*)$;
\item gr$_I(t_1 = t_2)$ is $\top$ if $t_1^I = t_2^I$ and $\bot$ otherwise;
\item gr$_I(F \odot G)$ is gr$_I(F) \odot$ gr$_I(G)$, where $\odot \in \{\land, 
\lor, \rar\}$;   
\item gr$_I(\forall x F(x))$ is \{gr$_I(F^x_{u^*}) | u \in 
|I|\}^\land$;
\item gr$_I(\exists x F(x))$ is \{gr$_I(F^x_{u^*}) | u \in |I|\}^\lor$.   
\end{itemize}
(By $F^x_{u^*}$ we denote the result of substituting $u^*$ 
for all free occurrences of $x$ in~$F$.)   
It is clear that 
for any first-order sentence $F$ over signature $\Sigma$, gr$_I(F)$ is an infinitary 
formula over the signature $A_{\Sigma, I}$. 

\medskip

\noindent {\it Example (continued)} 
\newline \noindent If $\Sigma$ consists of the unary predicate 
$p$ and the  propositional symbol $q$, and $I$ is an 
interpretation of $\Sigma$ such that the domain~$|I|$ is the set of all ground 
terms $t$, then the grounding of~\eqref{eq:r1fo} w.r.t. $I$ 
is~\eqref{eq:r1if}. (To simplify notation we identify the name of each term~$t$ 
with~$t$.) 

\medskip

According to Theorem 5 from \cite{tru12}, if $F$ is a first-order
sentence and~$I$ is an interpretation, then~$I$ is a first-order stable model 
of~$F$ iff~$I^r$ is an infinitary stable model of gr$_I(F)$. The proposition below 
generalizes this result to the 
case of~$\bp$-stable models. 
By~$\bp^{I}$ we denote the 
atomic formulas in~$A_{\Sigma, I}$ built with 
predicates from~$\bp$. 

\medskip

\noindent {\it Example (continued)} 
\newline \noindent If~$\bp$ is~$p$ then~$\bp^{I}$ is the set 
of all atoms of the form~$p(t)$. 
\begin{prop}\label{thm:infpstable}
For any first-order sentence $F$ over $\Sigma$ 
and any tuple $\bp$ of distinct predicate symbols from~$\Sigma$, an 
interpretation~$I$ is a $\bp$-stable model of $F$ iff $I^r$ is a 
$\bp^{I}$-stable model of~gr$_I(F)$. 
\end{prop}

\noindent {\it Example (continued)} 
\newline \noindent Let $I$ be the interpretation that interprets $p$ as 
identically false and assigns the value $\top$ to $q$. Then $I^r$ is $\{q\}$.  
Let $J$ be an interpretation that satisfies at least one atomic formula $p(t)$ 
and assigns the value $\bot$ to $q$. Then $J^r$ is $\{p(t)\; |\; J \models 
p(t)\}$ (the same as $\mp$ from the previous section).  We saw in the 
previous section that $\{q\}$-stable models of~\eqref{eq:r1if} are~$\{q\}$
and any non-empty set of atoms of the form $p(t)$. In accordance with the 
proposition above, $I$ and $J$ are $\{q\}$-stable models of~\eqref{eq:r1fo}.  

\begin{proof}[Proof of Proposition~\ref{thm:infpstable}]
Consider a first-order sentence $F$ and list of distinct predicate symbols~$\bp.$ 
Let~$\mq$ be the set of all predicates occurring in 
$F$ but not in~$\bp$. Consider an 
interpretation $I$ of the signature of $F$. By Theorem~2 from 
\cite{fer09}, 
$I$ is a $\bp$-stable 
model of $F$ iff  
it is a stable model of
$$
F \land \bigwedge_{q \in \mq} \forall {\bf x}(q({\bf x}) \lor \neg q({\bf x})),
$$ 
where ${\bf x}$ is a list of distinct object variables the same length as the 
arity of $q$. By Theorem~5 from \cite{tru12},~$I$ is a stable model of the 
formula above iff~$I^r$ is a stable model of the grounding of this formula 
w.r.t.~$I$.  The grounding of the formula above w.r.t.~$I$ is  
\beq
{\text gr}_I(F) \land  \bigwedge_{q \in \mq \atop A \in q^I} 
\left ( A \lor \neg A \right ). 
\eeq{eq:grI} 
By Proposition \ref{thm:astable}, $I^r$ is a stable model of \eqref{eq:grI} 
iff it is a $\bp^{I}$-stable model of~gr$_I(F)$. 
\end{proof} 
 
\section{Review: First-Order Splitting Lemma}\label{sec:fosplitting}

The lemma presented in the next section of this note is a generalization of 
the splitting lemma from \cite{fer09a}.

In order to state that lemma, we first 
review the definition of the predicate 
dependency graph given in that paper. We say that an 
occurrence of a predicate symbol or a subformula in a first-order formula $F$ 
is {\it positive} if it occurs in the antecedent of an even number of 
implications and {\it strictly positive} if it occurs in the antecedent of no 
implication. An occurrence of a predicate constant is said to be {\it negated} 
if it belongs to a subformula of the form $\neg F$, and {\it nonnegated} 
otherwise. A {\it rule} of a first-order formula $F$ is a strictly positive 
occurrence of an implication in $F$. The {\it (positive) predicate dependency 
graph} of a first-order formula $F$ w.r.t. a list ${\bf p}$ of 
distinct predicates, denoted DG$_{\bp}[F]$ is the directed graph that 
\begin{itemize}
\item has all predicate symbols in ${\bp}$ as its vertices, and
\item has an edge from $p$ to $q$ if, for some rule $G \rar H$ of $F$, 
\begin{itemize}
\item $p$ has a strictly positive occurrence in $H$, and
\item $q$ has a positive nonnegated occurrence in $G$. 
\end{itemize}
\end{itemize}
 
We say that a partition\footnote{We understand a partition of $X$ to be a set of 
disjoint subsets (possibly empty) that cover~$X$.} $\{\bp_1, \bp_2\}$ of the 
vertices in a graph $G$ is {\it separable (on $G$)} if every strongly 
connected component of $G$ is a subset of either $\bp_1$ or~$\bp_2$. (Here, 
we identify the list ${\bf p}$ with the set of its members.) 

The following assertion is a reformulation of Version 1 of the splitting lemma 
from \cite{fer09a}.  

\medskip

\noindent {\it Splitting Lemma} 
\newline \noindent If $F$ is 
a first-order sentence and $\bp_1,\bp_2$ are lists of distinct predicate symbols 
such that the partition $\{\bp_1,\bp_2\}$ is  separable on DG$_{\bp_1\bp_2}[F]$ 
then~$I$ 
is a ${{\bp}_1{\bp}_2}$-stable model of $F$ iff it is 
both  
a ${\bp}_1$-stable model and a ${\bp}_2$-stable 
model of $F$.  

\section{Infinitary Splitting Lemma}\label{sec:ifsplitting}

The statement of the infinitary splitting lemma refers to the positive dependency
graph of an 
infinitary formula. As we will see, the vertices of this graph correspond to 
intensional atoms. This definition is similar to the definition of a predicate 
dependency graph given in \cite{fer07} and \cite{fer09a} and reviewed in the 
previous section. The concepts 
necessary to define the dependency graph of an infinitary formula are 
all straightforward extensions of the concepts used in the previous section to define the predicate 
dependency graph in the first-order case. However, 
because infinitary formulas are not syntactic structures, we have to define 
these concepts recursively. 

We define the set of {\it strictly positive atoms} of an infinitary 
formula~$F$, denoted P($F$), recursively, as follows:
\begin{itemize}
\item For every atom $p \in \sigma$, P($p$) is $\{p\}$; 
\item P($\mathcal{H}^\land$) is $\bigcup_{H \in 
\mathcal{H}}$ P($H$), and so is P($\mathcal{H}^\lor$); 
\item P($G \rar H$) is P($H$).
\end{itemize}
 
The set of {\it positive nonnegated atoms} and the set of 
{\it negative nonnegated atoms} of an infinitary formula $F$, denoted Pnn($F$) 
and Nnn($F$) respectively, were introduced in \cite{lif12a}. These sets are 
defined recursively as well:
\begin{itemize}
\item For every atom $p \in \sigma$, Pnn($p$) is $\{p\}$; 
\item Pnn($\mathcal{H}^\land$) is $\bigcup_{H \in 
\mathcal{H}}$ Pnn($H$), and so is Pnn($\mathcal{H}^\lor$); 
\item Pnn($G \rar H$) is $\emptyset$ if $H$ is $\bot$ and Nnn($G$) $\cup$ 
Pnn($H$) otherwise.  
\end{itemize}
and
\begin{itemize}
\item For every atom $p \in \sigma$, Nnn($p$) is $\emptyset$; 
\item Nnn($\mathcal{H}^\land$) is $\bigcup_{H \in 
\mathcal{H}}$ Nnn($H$), and so is Nnn($\mathcal{H}^\lor$); 
\item Nnn($G \rar H$) is $\emptyset$ if $H$ is $\bot$ and 
Pnn($G$) $\cup$ Nnn($H$) otherwise.  
\end{itemize}
The set of {\it rules} of an infinitary formula is defined 
as follows:
\begin{itemize}
\item The rules of $G \rar H$ are $G \rar H$ and all rules of $H$;
\item The rules of $\mathcal{H}^\land$ and $\mathcal{H}^\lor$ are the rules of 
all formulas in $\mathcal{H}$.  
\end{itemize}

\noindent {\it Example (continued)}
\newline \noindent
The set of positive nonnegated atoms in formula \eqref{eq:r1if} is 
the same as the set of strictly positive atoms: $\{q\}$. The only rule of  
formula \eqref{eq:r1if} is the formula itself. 

\medskip

For any infinitary formula $F$ the {\it (positive) dependency graph} of $F$ 
(relative to a set of atoms $\mathcal{A})$, 
denoted DG$_{\mathcal{A}}[F]$, 
is the directed graph, that 
\begin{itemize}
\item has all atoms in $\ma$ as its vertices, and
\item has an edge from $p$ to $q$ if, for some rule $G \rar H$ of $F$, 
\begin{itemize}
\item $p$ is an element of P($H$), and
\item $q$ is an element of Pnn($G$). 
\end{itemize}
\end{itemize}

The following statement appears to be a plausible counterpart to the splitting 
lemma reproduced in Section \ref{sec:fosplitting} for infinitary formulas:

\beq
\tag{$\ast$}
\parbox{\dimexpr\linewidth-4em}{%
If $F$ is 
an infinitary formula and $\mp_1,\mp_2$ are sets of atoms such 
that 
the partition  $\{\mp_1,\mp_2\}$ is separable on 
DG$_{\mp_1 \cup \mp_2}[F]$ 
then $I$ 
is a ${{\mathcal{P}}_1 \cup {\mathcal{P}}_2}$-stable model of $F$ iff it is 
both  
a ${{\mathcal{P}}_1}$-stable model and a ${{\mathcal{P}}_2}$-stable 
model of $F$.
  \strut
}
\eeq{quote:wrong} 

But this statement does not hold; 
in the case of 
infinitary formulas separability is not a sufficient condition to ensure 
splittability. Let $F$ be the infinitary conjunction 
$$
\bigwedge_{n} \left (p_{n+1} \rar p_n \right ), 
$$
where the conjunction extends over all integers $n$. 
Let $\mp$ be the set of all atoms $p_n$. Let $\mp_1$ be the set 
$\{p_n \;|\; n \text{ is even}\}$, and $\mp_2$ be the set  
$\{p_n \;|\; n \text{ is odd}\}.$ 
Then the partition $\{\mp_1, \mp_2\}$ is separable on DG$_\mp[F]$ (shown in 
Figure~\ref{fig:one}). Indeed, the strongly connected components of this graph 
are singletons.  If $I$ is the set of all atoms $p_n$ 
then the reduct of~$F$ w.r.t.~$I$ is~$F$
itself. It is easy to check that $I$ is a $\mp_1$-stable model as well as a 
$\mp_2$-stable model of~$F$, but is not $\mp$-stable.  
This counterexample shows that ($\ast$) is incorrect.  

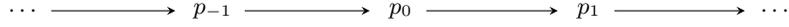
\begin{figure}\centering
\begin{tikzpicture}[shorten >=2pt, auto]
  \node[place] (n) {$p_0$};
  \node[place] (n-1) [left of=n] {$p_{-1}$};
  \node[place] (n+1) [right of=n] {$p_1$};
  \node[place] (n-2) [left of=n-1] {$\dots$};
  \node[place] (n+2) [right of=n+1] {$\dots$};
  \path[->]
     (n-2) edge[->] node[below] {} (n-1)
     (n-1) edge[->] node[below] {} (n)
     (n) edge[->] node[below] {} (n+1)
     (n+1) edge[->] node[below] {} (n+2);
\end{tikzpicture}
\caption{Any partition of the vertices in this graph is separable.}
\label{fig:one}
\end{figure}
In order to extend the splitting lemma to infinitary formulas, we will 
need a stronger notion of separability.  
An {\it infinite walk} $W$ of a directed graph $G$ is an infinite sequence $(v_1, v_2, \dots)$ of 
vertices occurring in $G$, such that each pair $v_i, v_{i+1}$ in $W$ corresponds 
to an edge in $G$. A partition $\{\mp_1, \mp_2\}$ of the vertices in $G$ will be 
called {\it infinitely separable (on~$G$)} if every infinite walk $(v_1, 
v_2, \dots)$ of $G$ visits either $\mp_1$ or $\mp_2$ finitely many times, 
that is either $\{i: v_i \in \mp_1\}$ or $\{i: v_i \in \mp_2\}$ is finite.  
\begin{prop}
For any graph $G$, 
\begin{enumerate}
\item[(i)] every infinitely separable partition of $G$ is separable, and
\item[(ii)] if $G$ has finitely many strongly connected components and partition 
$\{\mp_1, \mp_2\}$ is separable on $G$ then it is infinitely separable on $G$.
\end{enumerate}
\end{prop}
\begin{proof}
(i) We will prove the contrapositive: if $\{\mp_1, \mp_2\}$ is a partition that  
is not separable on $G$, then there is some strongly connected component of 
$G$ that contains at least one vertex from $\mp_1$ and at least one vertex 
from $\mp_2$. Let's call these vertices $v$ and $w$, respectively. 
Since $v$ and $w$ are in the same strongly connected component, each vertex is 
reachable from the other. Then there is an infinite walk that 
visits each of these vertices (and therefore both $\mp_1$ and $\mp_2$) 
infinitely many times, so that the partition is not infinitely separable on 
$G$. 
\smallskip
\newline\noindent (ii) Again we prove the contrapositive: 
if $\{\mp_1, \mp_2\}$ is  a partition that is not infinitely separable on $G$, 
then there is some infinite walk 
$(v_1, v_2, \dots )$ of $G$ that visits both~$\mp_1$ and~$\mp_2$ infinitely 
many times. Since there are only finitely many strongly connected components 
in $G$,  at least one strongly connected component of $\mp_1$ and 
at least one strongly connected component of $\mp_2$ must be visited 
infinitely many times. Call these strongly connected components $C_1$ and $C_2$ 
respectively; then~$C_1$ must be reachable from~$C_2$ and vice versa. Then $C_1 
= C_2$ so that the partition is not separable on $G$.
\end{proof}

Claim ($\ast$) will become correct if we require the partition $\{\mp_1, 
\mp_2\}$ to be infinitely separable: 

\medskip

\noindent {\it Infinitary Splitting Lemma}  
\newline \noindent If $F$ is 
an infinitary formula and $\mp_1,\mp_2$ are sets of atoms such 
that the partition $\{\mp_1,\mp_2\}$ is infinitely separable on 
DG$_{\mp_1 \cup \mp_2}[F]$ then $I$ 
is a    
${{\mathcal{P}}_1 \cup {\mathcal{P}}_2}$-stable model of $F$ iff it is 
both  
a ${{\mathcal{P}}_1}$-stable model and a ${{\mathcal{P}}_2}$-stable 
model of $F$.  

\medskip

The splitting lemma reproduced in Section \ref{sec:fosplitting} is a 
consequence of the infinitary 
splitting lemma in view of Theorem \ref{thm:infpstable} and the following fact:

\begin{prop}\label{prop:fo2inf}
For any first-order sentence $F$ and tuple $\bp$ of distinct predicate symbols, if 
$\{\bp_1, \bp_2\}$ is a partition of $\bp$ that is separable on DG$_\bp[F]$,  
then for any interpretation $I$, 
$\{\bp_1^I, \bp_2^I\}$  
is infinitely separable on DG$_{\bp^I}[{\text gr}_I(F)]$.
\end{prop}
\begin{proof}
If $\{\bp_1, \bp_2\}$ is a partition of $\bp$ that is separable on 
DG$_\bp[F]$,  then for any interpretation $I$, the partition $\{{\bf p}_1^I, 
{\bf p}_2^I\}$ is separable on the atomic dependency graph of gr$_I(F)$ with 
respect to $ {\bf p}^I$. Furthermore, it is easy to see that 
DG$_{\bp^I}[{\text gr}_I(F)]$  
must have finitely many strongly connected components, so that      
$\{{\bf p}_1^I, {\bf p}_2^I\}$ must be infinitely separable on it.  
\end{proof}

\section{Proof of the Infinitary Splitting Lemma}

The following two lemmas can be easily proved by induction on the rank of $F$. 

\begin{lemma}\label{lem:botreduct}
If $I$ does not satisfy $F$ then 
the reduct~$F^I$ is equivalent to $\bot$. 
\end{lemma}
\begin{lemma} \label{lem:nega} 
If the set $\ma$ is disjoint from P$(F)$ and $I$ satisfies $F$, then 
$I \setminus \ma$ satisfies~$F^I$. 
\end{lemma}

In particular, if $I$ satisfies $F$ then $I$ satisfies $F^I$. (This is the 
direction left-to-right of Proposition 1 from \cite{tru12}.)

Lemmas \ref{lem:2parts}--\ref{lem:loops} are similar to Lemmas 3--5 from 
\cite{fer09a}. 

\begin{lemma}\label{lem:2parts}
For any disjoint sets of atoms $\mb_1, \mb_2$, interpretation $I$, and 
formula $F$, 
\begin{enumerate}
\item[(i)] If $\mb_2$ is disjoint from Pnn$(F)$  
and $I \setminus\!\mb_1$ satisfies~$F^I$ then 
$I\setminus\!(\mb_1\!\cup\mb_2)$ satisfies~$F^I$.  
\item[(ii)] If $\mb_2$ is disjoint from  Nnn$(F)$
and $I \setminus\!(\mb_1\!\cup \mb_2)$ satisfies 
$F^I$ then $I \setminus\!\mb_1$ satisfies~$F^I$.  
\end{enumerate}
\end{lemma}
\begin{proof} Both parts of the lemma are proved simultaneously by induction on 
the rank of $F$. Here, we show only the most interesting case when $F$ is of 
the form $G \rar H$. (i) If $I$ does not satisfy $F$ the reduct is 
equivalent to $\bot$ so that the proposition is trivially 
true.  Assume that $I \setminus \mb_1$ satisfies $G^I \rar H^I$ and that $\mb_2$ 
is disjoint from Pnn$(G \rar H)$. 
Then either~$H$ is~$\bot$ or~$\mb_2$ is disjoint from both Nnn$(G)$ and 
Pnn$(H)$.
If~$H$ is~$\bot$ then the set P$(F)$ is empty, so that $(\mb_1  \cup \mb_2)$ is 
disjoint from it.  Then by Lemma~\ref{lem:nega}, if 
$I$ satisfies $F$ then  $I \setminus 
(\mb_1 \cup \mb_2)$ satisfies $F^I$. If, on the other hand, $\mb_2$ is disjoint 
from both Nnn$(G)$ and Pnn$(H)$, then by part 
(i) of the induction hypothesis we may conclude that 
\beq
\text{if } I \setminus \mb_1 \text{ satisfies }
H^I \text{ then so does } I \setminus (\mb_1 \cup \mb_2),
\eeq{eq:iha} 
and by part (ii) of 
the induction hypothesis we may conclude that  
\beq
\text{if } I \setminus (\mb_1 \cup  \mb_2) \text{ satisfies }
G^I \text{ then so does } I \setminus \mb_1.
\eeq{eq:ihb} 
Assume that $I \setminus (\mb_1 \cup 
\mb_2)$ satisfies $G^I$. Then by \eqref{eq:ihb}, $I \setminus \mb_1$ satisfies 
$G^I$. Then, since $I \setminus \mb_1$ satisfies $G^I \rar H^I$,  that 
interpretation must satisfy $H^I$. Then 
by \eqref{eq:iha} we can conclude that $I \setminus  (\mb_1 \cup \mb_2)$ 
satisfies $H^I$. It follows that that $I \setminus (\mb_1 \cup \mb_2)$ satisfies 
$G^I \rar H^I$. (ii) Similar to Part (i).  
\end{proof}

\begin{lemma} \label{lem:noedge}
Let $\mb, \mc$ be disjoint sets of atoms and let $F$ be an infinitary formula 
such that there are no 
edges from $\mb$ to $\mc$ in DG$_{\mb\cup\mc}[F]$. If $I \setminus 
(\mb\cup\mc)$ satisfies $F^I$ then so does $I \setminus \mb$.  
\end{lemma}
\begin{proof} The proof is by induction on the rank of $F$. Again we show 
only the most interesting case when $F$ is of the form $G~\rar~H$. Assume 
that $I \setminus (\mb \cup \mc)$ 
satisfies $(G~\rar~H)^I~=~G^I~\rar~H^I$. We need to show that $I \setminus 
\mb$ also satisfies $G^I \rar H^I$. If $\mb$ is disjoint from P$(H)$, then by 
Lemma~\ref{lem:nega}, $I \setminus \mb$ satisfies $H^I$, and therefore 
satisfies $G^I \rar H^I$. 
If, on the other hand, $\mb$ is not disjoint from P$(H)$ then $\mc$ must be 
disjoint from Pnn$(G)$, because there are no edges from $\mb$ to $\mc$ in
DG$_{\mb \cup \mc}[G \rar H]$.  
Then by Lemma~\ref{lem:2parts}(i), $I \setminus (\mb \cup \mc)$ satisfies 
$G^I$. Since we assumed that $I \setminus (\mb \cup \mc)$ satisfies 
$G^I \rar H^I$, it follows that $I \setminus (\mb \cup \mc)$ satisfies $H^I$. 
Since every edge in  DG$_{\mb \cup \mc}[H]$ occurs in  DG$_{\mb \cup 
\mc}[G \rar H]$ there is no edge from $\mb$ to $\mc$ in DG$_{\mb \cup \mc}[H]$. 
Then by the induction hypothesis, $I \setminus \mb$ satisfies~$H^I$ and 
therefore satisfies $G^I \rar H^I$.  
\end{proof}

\begin{lemma}\label{lem:loops}
For any non-empty graph $G$ and any infinitely separable partition $\{\ma_1, \ma_2\}$ on $G$, 
there exists a non-empty subset $\mb$ of the vertices in $G$
such that
\begin{enumerate}
\item[(i)] $\mb$ is either a subset of $\ma_1$ or a subset of $\ma_2$, and
\item[(ii)] there are no edges from $\mb$ to vertices not in $\mb$. 
\end{enumerate}
\end{lemma}
\begin{proof}
Since $\{\ma_1, \ma_2\}$ is infinitely separable on $G$, there is some vertex~$b$
 such that the set of vertices reachable from~$b$ is either a subset of $\ma_1$ 
or a subset of $\ma_2$.
(If no such $b$ existed then $\ma_1$ would be reachable from every vertex in 
$\ma_2$ and vice versa, and we could construct an infinite walk visiting 
both elements of the partition infinitely many times.)
It is easy to see that the set of all 
vertices reachable from $b$ satisfies both (i) and (ii).  
\end{proof}
\begin{proof}[Proof of the Infinitary Splitting Lemma]
Let $F$ be an 
infinitary formula such that the partition $\{\ma_1, \ma_2\}$ is infinitely separable on 
DG$_{\ma_1 \cup \ma_2}[F]$. We 
need to show that $I$ is an $\ma_1\!\cup\!\ma_2$-stable model of~$F$ iff it is 
an $\ma_1$-stable model and an $\ma_2$-stable model of $F$. The direction 
left-to-right is obvious. To establish the 
direction right-to-left, assume that $I$ is both an $\ma_1$-stable model and an 
$\ma_2$-stable model of $F$. 
By Proposition \ref{thm:astable} it is sufficient to show that $I$ is a 
minimal model of
\beq
F^I \land \bigwedge_{p \in I \setminus (\ma_1 \cup \ma_2)} p.
\eeq{eq:cupreduct2} 
Clearly, $I$ satisfies this formula. It remains to show
that $I$ is minimal.  Assume there is some non-empty subset~$X$ of~$I$ such 
that~$I \setminus X$ satisfies \eqref{eq:cupreduct2}. Then $I \setminus X$ 
satisfies the second conjunctive term of~\eqref{eq:cupreduct2}, so   
$ I \setminus (\ma_1 \cup \ma_2) \subseteq I \setminus X$. Consequently, $X 
\subseteq \ma_1 \cup \ma_2$.  
Consider the sets $X \cap \ma_1$ and 
$X \cap \ma_2$. 
Since $\ma_1$ and $\ma_2$ are infinitely separable on DG$_{\ma_1 \cup 
\ma_2}[F]$, the sets 
$X \cap \ma_1$ and  $X \cap \ma_2$ must be infinitely separable on DG$_X[F]$.  
Then by Lemma~\ref{lem:loops}, there is some non-empty set~$\mb$ that is 
either a subset of 
$X \cap \ma_1$ or a subset of $X \cap \ma_2$ and such that there are no edges 
from $\mb$ to $X \setminus \mb$. 
We will show that $I \setminus \mb$ satisfies
\beq
F^I \land \bigwedge_{p \in I \setminus \ma_1} p,
\eeq{eq:reda1}
which contradicts the assumption that $I$ is an $\ma_1$-stable model of $F$.
Since~$I~\setminus~X$ satisfies the first 
conjunctive term of \eqref{eq:reda1}, 
by~Lemma~\ref{lem:noedge} so does $I \setminus \mathcal{B}$. Assume, for 
instance, that 
$\mb$ is a subset of~$X \cap \ma_1$. Then~$\mb$ is a subset of $\ma_1$, so that 
$I \setminus \ma_1$ is a subset of $I \setminus \mb$. We may conclude that $I 
\setminus B$ satisfies the second conjunction term of \eqref{eq:reda1} as well.
\end{proof} 

\section{Infinitary Splitting Theorem}

The infinitary splitting lemma can be used to prove the following 
theorem, which is similar to the splitting theorem from \cite{fer09a}.  

\medskip

\noindent {\it Infinitary Splitting Theorem}  
\newline \noindent Let $F,G$ be infinitary formulas, and 
$\ma_1,\ma_2$ be disjoint 
sets of atoms such that 
the partition $\{\ma_1, \ma_2\}$ is infinitely separable on 
DG$_{\ma_1\!\cup\!\ma_2}[F \land G]$. If 
$\ma_2$ is disjoint from P$(F)$, and 
$\ma_1$ is disjoint from P$(G)$,
then for any interpretation $I$, $I$ is an $\ma_1\!\cup\!\ma_2$-stable model of 
$F\land G$ iff it is both an $\ma_1$-stable model of $F$ and an $\ma_2$-stable 
model of $G$.

\medskip

\noindent{\it Example (continued)}
\newline \noindent
Consider the conjunction of \eqref{eq:r1if} with the formula $\mp^\land$ 
where $\mp$ is as before some non-empty set of atoms of the form 
$p(t)$. We saw previously that 
$\{q\}$ and all non-empty sets of atoms of the form $p(t)$ are 
$\{q\}$-stable models of \eqref{eq:r1if}. 
It is easy to check that $\sigma\!\setminus\!\{q\}$-stable models of 
$\mp^\land$ are $\mp$ and $\mp \cup \{q\}$. In accordance with the 
splitting theorem, $\mp$ is the only stable model of this formula. 

\medskip

The following lemma, analogous to Theorem 3 from \cite{fer09}, is used to prove 
the infinitary splitting theorem. 

\begin{lemma}\label{lem:stabsat}
For any infinitary formulas $F,G$, if $\ma$ is disjoint from P$(G)$ then $I$ 
is an $\ma$-stable model of $F \land G$ iff it is an $\ma$-stable model of $F$ 
and satisfies $G$.  
\end{lemma}

\begin{proof} $\Leftarrow$: Assume $I$ is an $\ma$-stable model of $F$ and $I$ 
satisfies $G$. 
Since $I$ satisfies $G$ it satisfies~$G^I$. 
Since $I$ is an $\ma$-stable model of $F$, it is a minimal 
w.r.t. $\leq_{\ma}$ among the models of $F$, and consequently among 
the models of $F \land G$. 

\medskip
\noindent $\Rightarrow$: Assume $I$ is an $\ma$-stable model of $F \land G$. 
Then $I$ is a minimal model of~$(F \land G)^I$ w.r.t. $\leq_\ma$. 
So $I$ satisfies $F \land G$ and therefore satisfies $G$. It remains to show 
that there is no proper subset $J$ of $I$ such that $I \setminus J \subseteq 
\ma$ and $J$ satisfies $F^I$. Assume that there is some such $J$. Then $J$ 
must not satisfy $G^I$. (If it did, then~$I$ would not be minimal with respect 
to $\leq_\ma$ among the models of $(F \land G)^I$.) Let $\ma'$ denote 
$I \setminus J$.  Since $\ma$ is disjoint from P$(G)$, so is $\ma'$. 
So by Lemma \ref{lem:nega}, $I \setminus \ma' = J$ must satisfy $G^I$. 
Contradiction.  
\end{proof}

\begin{proof}[Proof of the Infinitary Splitting Theorem]
Let $F,G$ be infinitary formulas and let $\ma_1, \ma_2$ be disjoint sets of 
atoms such that the partition $\{\ma_1, \ma_2\}$ is infinitely 
separable on DG$_{\ma_1 \cup \ma_2}[F \land G]$ and the other conditions of the infinitary splitting theorem hold. 
By the infinitary splitting lemma, $I$ is an 
$\ma_1\!\cup\!\ma_2$-stable model of $F 
\land G$ iff it is both an $\ma_1$-stable model and an $\ma_2$-stable model of 
$F \land G$. Since~$\ma_2$ is disjoint from P$(F)$, by Lemma \ref{lem:stabsat}, 
$I$ is an $\ma_2$-stable model of $F \land G$ iff it is an 
$\ma_2$-stable model of $G$ 
and it satisfies $F$. Similarly, $I$ is 
an $\ma_1$-stable model of $F \land G$ iff it is an $\ma_1$-stable model of $F$ 
and it satisfies $G$. It remains to observe that if $I$ is an~$\ma_2$-stable 
model of $F$ then it satisfies $F$, and similarly if $I$ is an $\ma_1$-stable 
model of~$G$.  
\end{proof}


\section{Application: Infinitary Definitions}
About a formula $G$ and a set $\mq$ of atoms we will say that $G$ is a 
{\it definition for $\mq$} if it is a conjunction of a set of 
formulas of the form $H \land \mc^{\land} \rar q$, where $q$ is an atom
in $\mq$, $C$ is a  subset of $\mq$ (possibly empty), and no atoms 
from $\mq$ occur in $H$.\footnote{The relation 
$p$ {\it occurs in} $F$ is defined recursively in a straightforward 
way.}

A simple special case is ``explicit definitions'': conjunctions of formulas
$H \rar q$ such that atoms from  $\mq$ don't occur in any $H$. For example, 
\eqref{eq:r1if} is an explicit definition of $\{q\}$. The conjunction
of the formulas
$$
p_{\alpha\beta} \rar q_{\alpha\beta}\quad\hbox{and}\quad
q_{\alpha\beta} \land q_{\beta\gamma} \rar q_{\alpha\gamma}
$$
for all $\alpha,\beta,\gamma$ from some set of indices, which represents the 
usual recursive definition of transitive closure, is a definition in our sense
as well.  On the other hand, the formula $\neg q \rar q$ is not a definition.

The following theorem shows that all definitions are ``conservative''.  

\medskip
\noindent{\it Theorem on Infinitary Definitions}
\newline \noindent For any infinitary formula~$F$, any set~$\mq$ of 
atoms that do not occur in $F$, and any definition~$G$ for~$\mq$, the map
$I \mapsto I \setminus \mq$ is a 1-1 correspondence between the stable models
of $F \land G$ and the stable models of $F$. 

\medskip

This theorem generalizes the lemma on explicit definitions due to Ferraris
(\citeyear{fer05}) in two ways: it applies to infinitary formulas, and it
allows definitions to be recursive.

%

\begin{lemma}\label{lem:defproof}
If all atoms that occur in $F$ belong to $\ma$ then, for any interpretation~$I$,
$I$ is an $\ma$-stable model of $F$ iff $I \cap \ma$ is a stable model of $F$. 
\end{lemma}
\begin{proof}
If all atoms that occur in $F$ belong to $\ma$ then
$$F^{I \cap \ma} \land \bigwedge_{p \in I \setminus (I\cap\ma)} p$$
is identical to~\eqref{eq:modred}.
\end{proof}
\begin{lemma}\label{lem:qdefhelp}
Let $G$ be a definition for a set $\mq$ of atoms, and let $I$ be a model
of $G$. For any 
subset $K$ of $I$ such that $K \setminus \mq = I \setminus \mq$, $K$ satisfies 
$G^I$ iff $K$ satisfies $G$.
\end{lemma}

\noindent {\it Proof.}
\newline \noindent We can show that~$K$ satisfies a conjunctive term
$H \land \mc^{\land} \rar q$ of $G$ iff $K$ satisfies its reduct
$H^I \land (\mc^{\land})^I \rar q^I$ as follows:
\begin{align*}
&K \not \models H^I \land (C^{\land})^I \rar q^I\\
\text{iff } &K \models H^I,\ K\models (C^{\land})^I, \text{ and } K \not \models q^I\\ 
\text{iff } &K \models H^I,\ K\models (C^{\land})^I, \text{ and } q \not \in K
&&(\text{because }K\subseteq I)\\
\text{iff } &I \models H^I,\ K\models (C^{\land})^I, \text{ and } q \not \in K
&&(\text{$K$ and $I$ agree on atoms occurring in~$H$})\\
\text{iff } &I \models H,\ K\models (C^{\land})^I, \text{ and } q \not \in K\\
\text{iff } &K \models H,\ K\models (C^{\land})^I, \text{ and } q \not \in K&&
(\text{$K$ and $I$ agree on atoms occurring in~$H$})\\
\text{iff } &K \models H,\ C\subseteq K, \text{ and } q \not \in K
&&(K\subseteq I)\\ 
\text{iff } &K \not \models H \land C^{\land} \rar q. \qquad \Box 
\end{align*}

\begin{lemma}\label{lem:qdef}
Let $G$ be a definition for a set $\mq$ of atoms. For any set $J$ of atoms
disjoint from  $\mq$ there exists a unique $\mq$-stable model $I$ of $G$ such
that $I \setminus \mq = J$.  
\end{lemma}
\begin{proof}
Let $I$ be the intersection of all models $K$ of $G$ such that $K\setminus\mq = J$.
We will show first that $I$ satisfies~$G$.  Assume otherwise, and take
a conjunctive term $H \land C^{\land} \rar q$ of~$G$ that is not satisfied by $I$. Then 
$I$ satisfies~$H$, $C\subseteq I$, and $q\not\in I$.  By the choice of~$I$, it
follows that there is a model~$K$ of~$G$ such that $K \setminus Q = J$ and
$q\not\in K$.  On the other hand, since~$I$ satisfies~$H$ and does not differ from~$K$
on atoms occurring in~$H$,~$K$ satisfies~$H$.  Since $C\subseteq I\subseteq K$,
$K$ satisfies $C^{\land}$.  Hence $K$ does not satisfy one of the conjunctive terms
of~$G$, which is a contradiction.
Thus~$I$ is a model of~$G$, and consequently a model of~$G^I$.  To prove that it is
$\mq$-stable, consider any model~$K$ of~$G^I$ such that $K\leq_\mq I$.  By
Lemma \ref{lem:qdefhelp}, $K$ is also a model $G$.  By the choice of~$I$, it follows
that $I\subseteq K$. Consequently $K=I$.

It remains to show that $I$ is unique. Let $K$ be a $\mq$-stable model of $G$ 
such that $K \setminus Q = J$. 
It is easy to see 
that $I \subseteq K$. Furthermore, $K$ satisfies $G^K$ and $I$ satisfies $G$, so 
by Lemma \ref{lem:qdefhelp}, $I$ satisfies $G^K$. Since $I \leq_\mq K$, it 
follows that $I = K$.   
\end{proof}

\begin{proof}[Proof of Theorem on Infinitary Definitions]
Let $\sigma$ denote the set of all atoms occurring in $F \land G$. Since
atoms from $\mq$ do not occur in $F$ and P$(G)\subseteq\mq$, there are no
edges from $\sigma\!\setminus\!\mq$ to $\mq$ in DG$_ {\sigma}[F \land G].$ 
Consequently the partition $\{\sigma\!\setminus\!\mq, \mq\}$ is infinitely 
separable on this graph. 
By the splitting theorem for infinitary formulas, an interpretation $I$ is a 
stable model of $F \land G$ iff it is a
$(\sigma\!\setminus\!\mq)$-stable model of $F$ and a $\mq$-stable model of~$G$.
Consider a stable model~$I$ of $F \land G$.  
We have seen that~$I$ is a $(\sigma\!\setminus\!\mq)$-stable model of~$F$.
By Lemma~\ref{lem:defproof}, it follows that~$I \setminus \mq$ is a  
stable model of $F$.
Consider now a stable model $J$ of $F$, and let~$S$ be the set of all
interpretations~$I$ such that $J  = I \setminus \mq$.   We will show that~$S$
contains exactly one stable model of $F \land G$, or equivalently,  
that~there is exactly one interpretation that is
a $(\sigma\!\setminus\!\mq)$-stable model of $F$ and a $\mq$-stable model of $G$ 
in $S$.  
By Lemma \ref{lem:defproof}, any interpretation in~$S$ is a
$(\sigma\!\setminus\!\mq)$-stable model of $F$.  By Lemma \ref{lem:qdef},~$S$
contains exactly one $\mq$-stable model of~$G$.  
\end{proof}

\section{Conclusion}

In this note, we defined and studied stable models for infinitary 
propositional formulas with extensional atoms. The use of extensional atoms 
facilitates a more modular view of logic programs, as evidenced by the Theorem 
on Infinitary Definitions. The proof of this theorem relies on the Splitting 
Theorem, and the proof of that theorem makes critical use of the distinction 
between intensional and extensional atoms. 

\section*{Acknowledgements}

Many thanks to Yuliya Lierler, Dhananjay Raju, and the anonymous referees for 
useful comments. Both authors were partially supported by the National Science 
Foundation under Grant IIS-1422455. 

\bibliographystyle{acmtrans}
\bibliography{bib}

\end{document}